\documentclass[runningheads,envcountsect,10pt]{llncs}
\usepackage{amsmath}
\usepackage{mathrsfs}

\usepackage{multirow}
\usepackage{amsfonts}
\usepackage[colorlinks]{hyperref}
\usepackage{setspace}
\usepackage{float}
\usepackage{textcomp}
\usepackage{enumerate}
\usepackage[ruled,linesnumbered]{algorithm2e}
\usepackage{subcaption}
\usepackage{marvosym}

\usepackage{graphicx}
\DeclareGraphicsExtensions{.pdf,.jpg,.png}
 
\title{Turing Machines with Two-level Memory: A Deep Look into the Input/Output Complexity\thanks{This work was supported by the National Natural Science Foundation of China under grants 61832003, 61972110 and U1811461.}}
\titlerunning{TM-TLM: A Deep Look into the IO Complexity}

\author{Hengzhao Ma\inst{\dag}\textsuperscript{\Letter} \and Jianzhong Li\inst{\dag}\textsuperscript{\Letter} \and Xiangyu Gao\inst{\ddag\dag} \and Tianpeng Gao\inst{\ddag\dag}
}
\authorrunning{H. Ma, J. Li, X. Gao, and T. Gao}
\institute{
	\dag\;Shenzhen Institutes of Advanced Technology, Chinese Acadamy of Sciences\\
	\ddag\; Harbin Institute of Technology\\
	\Letter\; \{hz.ma,lijzh\}@siat.ac.cn
}

\begin{document}

\maketitle

\begin{abstract}
The input/output complexity, which is the complexity of data exchange between the main memory and the external memory, has been elaborately studied by a lot of former researchers. However, the existing works failed to consider the input/output complexity in a computation model point of view. In this paper we remedy this by proposing three variants of Turing machine that include external memory and the mechanism of exchanging data between main memory and external memory. Based on these new models, the input/output complexity is deeply studied. We discussed the relationship between input/output complexity and the other complexity measures such as time complexity and parameterized complexity, which is not considered by former researchers. We also define the external access trace complexity, which reflects the physical behavior of magnetic disks and gives a theoretical evidence of IO-efficient algorithms. 
\end{abstract}

\section{Introduction}
The concept of \textit{data intensive computing} originates in 1980s \cite{alexander1988process,copeland1988data}, and has been a hot research field since then \cite{gorton2008data,johnston1998high,kouzes2009changing}. As we are entering the era of big data, it is becoming more and more common to deal with data up to petabytes in many research fields, such as artificial intelligence \cite{raju2020big}, bio-informatics \cite{elworth2020petabytes}, data warehouse \cite{pandis2021evolution}, and so on, which puts even more importance on data intensive computing. In such data intensive computing tasks, the storage and transportation of the data often become the inevitable bottleneck. The reason is that the massive data must be stored in external memory such as hard disks or solid state disks, but the speed of the external memory is usually one or two order of magnitudes slower than that of CPU and main memory. Henceforth, it is important to consider the complexity of data exchanging between main memory and external memory in data intensive computing, which is the well-known input/output complexity. Throughout this paper we will use IO for abbreviation of input/output.

The \textit{computational models} are the canonical tool to study the computational complexity. A computational model formally defines how the computation proceeds on it, and the computational complexity is defined based on different aspects of the computation procedure. For example, the serial time and space complexity is based on the Turing machine, and the parallel time complexity is based on the PRAM model and logic circuit model. However, all classical computational models do not have the ability to model the IO operations and thus do not support the analysis of IO complexity. For example, the classical Turing machine only models the main memory computation. There is no representation of the external memory nor the mechanism of exchanging data between main memory and external memory in Turing machine. 

To solve the disadvantages of classical computational models for analyzing IO operation, several models with multi-level of memory have been proposed, such as the hierarchical memory model \cite{aggarwal1987model} and uniform memory hierarchy model \cite{alpern1994uniform}. However, these models are cost models rather than computational models, which do not formally define how the computation is executed on them. These models only focus on defining the cost parameters of specified aspects of the model, and the goal is to calculate the total cost of running an algorithm on it.  Taking the hierarchical memory model \cite{aggarwal1987model} as an example, it only defines that the cost of accessing a memory location is proportional to the length of the address. Under the parameters defined in \cite{aggarwal1987model}, the Fast Fourier Transform algorithm which has $O(n\log{n})$ time complexity on RAM model, will have $O(n\log{n}\log{\log{n}})$ time complexity. In such a sense, the complexity calculated  under these cost models is a total cost of main memory computation and IO operation, which can not separate the time complexity and IO complexity. However, understanding the IO complexity is more important in data intensive computing, but these cost models \cite{aggarwal1987model,alpern1994uniform} fail to do so.

The most well-acknowledged model for analyzing the IO complexity is proposed in \cite{aggarwal1988input}, which consists of four parameters $N,M,B$ and $P$. 
$N$ is the number of records in the input, $M$ is the size of the main memory, 
$B$ is the block size, and $P$ is the number of blocks that can be transferred concurrently. 
The IO complexity is considered as the number of IO operations to fulfill the computation task. Though well-known and well-studied, the model in \cite{aggarwal1988input} does not specify how the computation is executed on it like most of the cost models. Only the mechanism of data exchange between main memory and external memory is explicitly defined. Therefore, it focuses only on the IO complexity, but looses the insight on the relationship between the IO complexity and other complexities such as the time and space complexity.

To remedy the defects of existing models for analyzing the IO complexity, this paper proposes new computational models that can accurately describe the computation procedure involving main memory computation and IO operation, and thus can analyze the time, space and IO complexity simultaneously. 
Specifically, three new variants of Turing machine are proposed in this paper. Based on these new models, we also study the IO complexity in multiple aspects. 

(1) The first model, the Turing machine with two-level memory (TM-TLM), generalizes the Turing machine by equipping it with external memory and the ability to exchange data between main memory and external memory. The time, space and IO complexity of TM-TLM is defined, and the relationship between IO complexity and other kinds of complexity is studied, including the time complexity and parameterized complexity.

(2) The second model is the Random Access Turing machine with two-level memory (RATM-TLM). It is a generalization of the RATM \cite{gao2020recognizing} model which is designed to support sub-linear time computation. The sub-linear IO complexity is  discussed based on RATM-TLM. 

(3) Finally the Random Access Turing machine with Blocking-IO (RATM-BIO) is proposed, which explicitly models the cost of retrieving data on external memory and reflects the behavior of hard disks. Based on RATM-BIO, the external memory access trace complexity is defined, which models the cost of retrieving data on external memory and gives an theoretical point of view of IO-efficient algorithms.

The three proposed models have different dedicated usage. The TM-TLM is the basic model since it directly generalizes the classical Turing machine. It is easier to use when analyzing the IO complexity and time complexity, and most of the results in this paper are based on TM-TLM. However, the TM-TLM can not be used to study the sub-linear time and IO complexity, and this is where the RATM-TLM should be used since it has the power of random access. Both TM-TLM and RATM-TLM consider the IO operations as special oracles and assume that an IO operation takes one unit of time. This simplifies the analysis of IO complexity, but looses the details to reflect the behavior of realistic external memory. Therefore, the RATM-BIO explicitly models the behavior of the external memory, where the  pattern of external memory access can influence the cost of external memory access. In this way, RATM-BIO provides the ability to analyze the IO-efficient algorithms.

The rest of the paper is organized as follows. We first go over the related works in Section \ref{subsec:rwork}.
The TM-TLM is defined in Section \ref{sec:tm-tlm}, and the relationship between IO complexity and other complexity measures such as time complexity and parameterized complexity are discussed. In Section \ref{sec:ratm-tlm}, the RATM-TLM is defined, and the sub-linear time and IO complexity is studied. Section \ref{sec:RATM-BIO} defines the RATM-BIO, and discusses the external access trace complexity. Section \ref{sec:usage} discusses the usage of the three models by raising some concrete examples. Finally Section \ref{sec:conc} concludes the paper.

\subsection{Related Works}\label{subsec:rwork}
\subsubsection{Two-level memory model and IO complexity.}
We have mentioned the model proposed in \cite{aggarwal1988input}, and this model is later generalized as the Parallel Disk Model (PDM). There are two variants of PDM. One assumes that there are $D$ channels inside a single disk that can transfer data simultaneously, and the other considers that there are $D$ independent disks working together to serve one CPU. Armen et. al. \cite{armen1996bounds} proved that the power of the former one is strictly stronger than the latter one. The PDM model has been the standard model to analyze IO complexity, and the IO complexity of a variety of problems have been studied. The IO complexity of four basic operations on external memory is studied in \cite{vitter2001external}, including scan, search, sort and output. The researchers also studied the IO complexity of more complicated problems, such as triangle enumeration \cite{pagh2014triangle}, transitive closure \cite{ullman1991input}, matrix multiplication \cite{pagh2014matrix} and so on. As we mentioned before, the PDM model focuses only on the IO complexity, but can not analyze the IO complexity and other complexities simultaneously.

\subsubsection{Models with multi-level memory.}
Although classical computational models rarely consider the external memory, there exist many cost models that do. Some of them even consider a hierarchy of memories, reflecting the modern computer architecture which consists of multi-level cache, main memory and external memory. There are several models that consider multi-level memory, but they differ in the way of modeling the cost of accessing different level of memory.

The Hierarchical Memory Model \cite{aggarwal1987model} is an early work in this direction. It assumes $k$ levels of memory, each containing $2^k$ locations, and access to a location $x$ takes $\lceil \log{x}\rceil$ time. A later model, the Uniform Memory Hierarchy model (UMH) \cite{alpern1994uniform}, relates the cost of memory access to the memory level number, not the memory address. It also includes several other parameters, such as the block size and the bandwidth of each level of the memory, which increases the hardness to analyze the complexity of algorithms running on this model. Another model $DRAM(h,k)$ \cite{zhang2003dram} is a parallelized model which considers $k$ threads cooperating on $h$ levels of shared memory. Unlike the UMH model, where the cost of accessing data in the same level of memory is the same, $DRAM(h,k)$ considers that the cost of memory access is influenced by the memory access pattern, such as temporal and spatial locality, contiguous and non-contiguous accesses. Thus, different implementations of an algorithm may have different memory access cost on $DRAM(h,k)$ model.

These models with multi-level memory provide details that represent the realistic architecture of modern computers. However, more details make the model more complicated, and make it more difficult to use these models to analyze the performance of algorithms. Actually, two-level memory is enough to analyze the IO complexity.

\subsubsection{IO-efficient algorithms.}
In most of the above models, the IO complexity is modeled as the number of IO operations multiplying the cost of a single IO operation, and the cost of each IO operation is assumed to be the same. However, as pointed out in \cite{zhang2003dram}, the behavior of the external memory is radically different with that of the main memory, where the cost of an IO operation may be significantly affected by the access pattern. For example, reading or writing data on contiguous blocks in external memory is a lot faster than reading or writing data on randomly distributed blocks \cite{ruemmler1992unix,simitci1998comparison}.  
Thus, it is very important to design algorithms that is IO-efficient by taking the characteristic of external memory into consideration. A widely used technique for IO-efficient algorithm is to turn a set of non-contiguous accesses into contiguous access in a batched manner, e.g., the Log Structure Merge (LSM) tree \cite{o1996log}. See \cite{maheshwari2003survey} for a good survey of IO-efficient algorithms.

\section{Turing Machine with two-level memory}\label{sec:tm-tlm}

\subsection{Definition of Turing Machine with two-level memory}
As shown in Figure \ref{fig:tm-tlm}, a Turing Machine with two-level memory, TM-TLM for short, has 3-tapes. The first one is the limited main memory tape, the second one is the unlimited external memory tape, and the third one is the address tape for the external memory. 
The IO operations of TM-TLM are modeled as two oracles that can execute the \textit{Read} and \textit{Write} operations in the following way. 
There are two sets of special states in TM-TLM that are \textit{Read} states and \textit{Write} states. When TM-TLM enters a \textit{Read} state, it will write an address $addr$ to the address tape, and the content of the cell in the main memory tape where the head of the main memory tape points to will be replaced by the content at address $addr$ in the external memory tape. For a \textit{Write} state, the TM-TLM will write an address $addr$ to the address tape, and the content of the cell at address $addr$ in the external tape will be replaced by the content of the cell in the main memory tape where the head of the main memory tape points to. 
The \textit{Read} and \textit{Write} operations are included in the transition functions of TM-TLM. 
The head can shift left or right or stay in the current cell after a \textit{Read} or \textit{Write} operation. It is assumed that the \textit{Read} and \textit{Write} operations are executed in one unit of time, and they can be regarded as two special oracles.
TM-TLM is formally defined in the following Definition \ref{def:tm-tlm}.

\begin{figure}[H]
	\centering
	\includegraphics[width=0.8\linewidth]{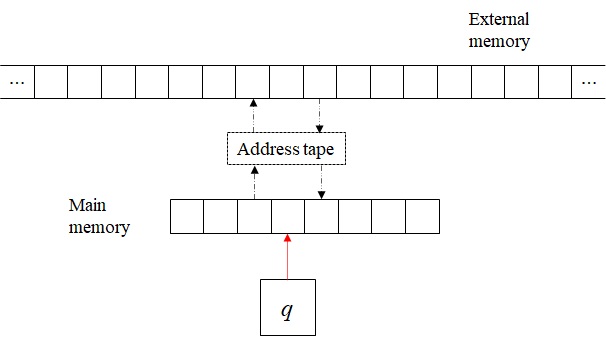}
	\caption{TM-TLM}\label{fig:tm-tlm}
\end{figure}

\begin{definition}\label{def:tm-tlm}
	A TM-TLM is a system $\{Q,\Sigma, \Gamma, M, \delta, Read, Write, B,q_r,q_w,q_f\}$ where \\
	$Q$ is the finite set of states,\\
	$\Sigma$ is the finite set of input symbols,\\
	$\Gamma\supseteq \Sigma$ is the finite set of tape symbols,\\
	$M$ is the size of the finite main memory tape,\\
	$\delta$ is the transition function: $Q\times \Gamma \rightarrow  
	Q\times\Gamma \times\{L,S,R\}$,\\
	$Read\subset \delta$ is the read operations described in the fist paragraph of this section,\\ 
	$Write\subset \delta$ is the write operations described in the fist paragraph of this section,\\
	$B\in \Gamma\setminus \Sigma$ is the blank symbol,\\
	$q_r\subset Q$ is the read states defined in the fist paragraph of this section, \\
	$q_w\subset Q$ is the write states defined in the fist paragraph of this section, and\\
	$q_f\in Q$ is the accepting state.\\
	
\end{definition}

Next we define the complexity measures of TM-TLM. Since the transitions of a TM-TLM include IO operations, it is necessary to consider the IO complexity of a TM-TLM.

\begin{definition}[Time Complexity]
	Given a TM-TLM $\mathcal{M}$ and an input string $x$, the time of $\mathcal{M}$ on $x$, denoted as $T_{\mathcal{M}}(x)$, is the number of transition functions except for $Read$ and $Write$ operations executed during the computation of $\mathcal{M}$ on $x$. Further, if  $T_{\mathcal{M}}(x)=O(T(n))$ for almost all $n\in Z^+$ and any $|x|=n$, we say that the time complexity of $\mathcal{M}$ is $O(T(n))$.
\end{definition}

\begin{definition}[IO Complexity]
	Given a TM-TLM $\mathcal{M}$ and an input string $x$, the IO time of $\mathcal{M}$ on $x$, denoted as $IO_{\mathcal{M}}(x)$, is the number of $Read$ and $Write$ operations executed during the computation of $\mathcal{M}$ on $x$. Further, if  $IO_{\mathcal{M}}(x)=O(IO(n))$ for almost all $n\in Z^+$ and any $|x|=n$, we say that the IO complexity of $\mathcal{M}$ is $O(IO(n))$.
\end{definition}

\begin{definition}[Space Complexity]
	Given a TM-TLM $\mathcal{M}$ and an input string $x$, the space of $\mathcal{M}$ on $x$, denoted as $S_{\mathcal{M}}(x)$, is the number of cells used on the external tape during the computation of $\mathcal{M}$ on $x$. Further, if  $S_{\mathcal{M}}(x)=O(S(n))$ for almost all $n\in Z^+$ and any $|x|=n$, we say that the space complexity of $\mathcal{M}$ is $O(S(n))$.
\end{definition}

Note that the space complexity of TM-TLM is actually the external space complexity. The main memory complexity is a fixed parameter of TM-TLM, since the size of the main memory is $M$.
	
\begin{definition}
	A $M$-memory $(T(n),IO(n))$-time $S(n)$-space TM-TLM $\mathcal{M}$ is a TM-TLM such that, the size of the main memory of $\mathcal{M}$ is $M$, the time complexity of $\mathcal{M}$ is $O(T(n))$, the IO complexity of $\mathcal{M}$ is $O(IO(n))$, and the space complexity of $\mathcal{M}$ is $O(S(n))$.
\end{definition}

Similar with the classical Turing machine, the time and IO complexity is usually more important than the space complexity for TM-TLM, and thus the space complexity is usually omitted to describe a TM-TLM. 
%From now on, if a function $f$ is computable by a $M$-memory $(T(n),IO(n))$-time TM-TLM $\mathcal{M}$, we also say that $f$ is computable by $\mathcal{M}$ in $(T(n),IO(n))$-time. 
Besides, the time and IO complexity of a $M$-memory TM-TLM may include $M$ as a parameter. For simplification, we will omit the parameter $M$ unless necessary in the rest of the paper. Throughout this paper it is assumed that $M<n$ since IO complexity will be meaningless if $M \geq n$.

\subsection{TM-TLM v.s. Turing Machine}
\begin{lemma}\label{lema:1-memory-simulation}
	If a function is computable by a $M$-memory $(T(n),IO(n))$-time TM-TLM, then there is a $1$-memory $(T(n),T(n)+IO(n))$-time TM-TLM that can compute $f$.
\end{lemma}

\begin{proof}
	Given a $M$-memory TM-TLM $\mathcal{M}$, the $1$-memory TM-TLM $\mathcal{M}'$ simulates $\mathcal{M}$ as follows. We use $M$ cells on the external tape of $\mathcal{M}'$ to simulate the main memory of $\mathcal{M}$. We call the $M$ cells as main memory mapping cells, and the other cells as external memory mapping cells. The main memory mapping cells are put on the left of the leftmost cell of external memory mapping cells, and a splitter symbol $\#$ is used to split the two parts, where $\#$ is not used anywhere else. The only main memory cell of $\mathcal{M}'$ is used to transfer data between the main memory mapping cells and external memory mapping cells. $\mathcal{M}'$ also has an register to remember the position of the main memory cell where the head of $\mathcal{M}$ points to.
	
	For each transition of $\mathcal{M}$, $\mathcal{M}'$ simulates it by two IO operations and one transition. Since the transition of $\mathcal{M}$ affects one main memory cell, the two IO operations of $\mathcal{M}'$ manipulates the affected main memory mapping cell accordingly. The transition of $\mathcal{M}'$ changes the register to conform with the position of the head of $\mathcal{M}$.
	
	For each IO operation of $\mathcal{M}$, $\mathcal{M}'$ simulates it by two IO operations. For read operation of $\mathcal{M}$, $\mathcal{M}'$ reads from the external mapping cell and writes to the main memory mapping cell. For write operation of $\mathcal{M}$, $\mathcal{M}'$ reads from main memory mapping cell and writes to external memory mapping cell.
	
	The correctness of the above simulation can be easily verified, and the time complexity of the simulation is straightforward.
	\qed
\end{proof}

\begin{theorem}\label{thrm:tlm-sim-tm}
	If a function $f$ is computable by a $T(n)$-time Turing machine $\mathcal{M}$, then there is a $(T(n),T(n))$-time TM-TLM $\mathcal{M}'$ that can compute $f$.
\end{theorem}

\begin{proof}
	A 1-memory TM-TLM $\mathcal{M}'$ suffices to simulate the Turing machine. The external memory of $\mathcal{M}'$ simulates the working tape of the Turing machine $\mathcal{M}$. $\mathcal{M}'$ also needs a register to remember the position on the working tape where the head points to. Then similar with the proof of Lemma \ref{lema:1-memory-simulation}, for each transition of $\mathcal{M}$, $\mathcal{M}'$ simulates it by two IO operations and one transition. The two IO operations changes the cell affected by the transition of $\mathcal{M}$, and the transition of $\mathcal{M}$ changes the value of the register to remember the new position of the head of $\mathcal{M}$. In such way, $\mathcal{M}'$ can correctly simulate $\mathcal{M}$ in $(T(n),T(n))$-time.
\end{proof}

\begin{theorem}\label{thrm:tm-sim-tlm}
	If a function $f$ is computable by a $M$-memory $(T(n),IO(n))$-time $S(n)$-space TM-TLM $\mathcal{M}$,  then there is a $O(T(n)+IO(n)*S(n))$-time Turing machine $\mathcal{M}'$ that can compute $f$.
\end{theorem}

\begin{proof}
	We use a 3-tape Turing machine $\mathcal{M}'$ to simulate TM-TLM $\mathcal{M}$. The three tapes are denoted as tape-1, tape-2 and tape-3, respectively. Tape-1 of $\mathcal{M}'$ simulates the main memory tape of $\mathcal{M}$, and tape-2 and tape-3 of $\mathcal{M}'$ simulates the address tape and external tape of $\mathcal{M}$. For the main memory computations of TM-TLM $\mathcal{M}$, the Turing machine $\mathcal{M}'$ simulates them by doing exactly the same transitions on tape-1. For each IO operation of $\mathcal{M}$, the Turing machine $\mathcal{M}'$ simulates it as follows. First $\mathcal{M}'$ writes the address to tape-2 using $O(\log{S(n)})$ time, then the head on tape-3 moves to the designated cell using  $O(S(n))$ time. Therefore, the each IO operation of  $\mathcal{M}$ is simulated by $\mathcal{M}'$ in $O(\log{S(n)})+O(S(n))=O(S(n))$ time. In conclusion, the total time of simulation is $O(T(n)+IO(n)*S(n))$.
\end{proof}

\subsection{TM-TLM v.s. Block Transfer TM-TLM}
The realistic external memory such as magnetic disk has the ability of block transfer, i.e., a single read or write operation can transfer $B$ records simultaneously between main memory and external memory, where $B$ is the block size. 
The TM-TLM is modeled as single cell transfer which makes it easier to design algorithms on it. However, it must be proved that it does not reduce the computational power of TM-TLM by assuming single cell transfer. Here we prove that the block transfer and single cell transfer are equivalent under big-O notation.

Denote TM-TLM-BT as Turing machine with two level memory and block transfer ability. For ease of discussion, we assume that the size of the main memory $M$ can be divided by the block size $B$. The data in the main and external memory are aligned to the block boundary, and any IO operation will only affect a single block. Moreover, even if the IO operation asks for transferring a single cell, all the $B$ cells in the block will be transferred simultaneously.

\begin{theorem}\label{thrm:block-sim}
	If a function $f$ is computable by a $M$-memory $(T(n),IO(n))$-time TM-TLM-BT, then there is a $M$-memory $(T(n)+B\cdot IO(n),B\cdot IO(n))$-time TM-TLM that can compute $f$.
\end{theorem}

\begin{proof}
	Given the $M$-memory TM-TLM-BT $\mathcal{M}$, the $M$-memory TM-TLM $\mathcal{M}'$ simulates $\mathcal{M}$ as follows. $\mathcal{M}'$ executes the same move as $\mathcal{M}$ for every transition $\delta$ that is not IO operation. And for each IO operation of $\mathcal{M}$, $\mathcal{M}'$ simulates it by $B$ consecutive IO operations which transfer exactly the same $B$ cells between main and external memory. After transferring the data, the head of $\mathcal{M}'$ is moved to the same position as $\mathcal{M}$, using at most $B$ moves. In conclusion, for each IO operation of $\mathcal{M}$, $\mathcal{M}'$ uses $B$ IO operations and at most $B$ transitions to simulate. Thus the result in the lemma follows.	
	\qed
\end{proof}

\subsection{IO complexity v.s. Parameterized Complexity}\label{subsec:io-parameterized}
Since TM-TLM has a limited main memory of size $M$, there is an intuition that $M$ can be regard as a parameter, and a connection between IO complexity an parameterized complexity can be established.

\begin{lemma}\label{lema:tlm-compute-kclique}
	The $k$-clique problem can be solved by a $(k+1)$-memory TM-TLM in $(k^2n^k,k^2n^k)$-time.
\end{lemma}

\begin{proof}
	The TM-TLM simulates the basic $O(k^2n^k)$-time algorithm for $k$-clique as follows.  The input graph is stored on the external memory. For the main memory, the leading $k$ cells in the main memory is used to enumerate the id's of $k$ nodes, and the last cell is used to read the edges corresponding to the enumerated nodes and check if they form a clique. Note that here we implicitly use the tape-compression technique and assumes that the id of a node or an edge can be stored in one single cell. The number of IO operations is $O(k^2n^k)$ since the algorithm enumerates all possible $k$-set of nodes, and for each $k$-set of the nodes the algorithm reads $k^2$ edges. As to the main memory computation, for each possible $k$-set of nodes, it needs $O(k)$ time to list them on the main memory tape, and $O(k^2)$ time to generate the end points of the edges to read. Then the main memory computation time complexity is $O(k^2n^k)$.
	\qed
\end{proof}

\begin{definition}[Parameterized Reduction, see \cite{downey2013fundamentals}]
	Let $A,B\subseteq \sigma^*\times N$ be two parameterized problems. A parameterized reduction from $A$ to $B$ is an algorithm that, given an instance $(x,k)$ of $A$, outputs an instance $(x',k')$ of B such that: (1) $(x,k)$ is a Yes-intance of $A$ if and only if $(x',k')$ is a Yes-instance of $B$, (2) $k'\le g(k)$ for some computable function $g$, and (3) the running time of the reduction is bounded by $f(k)\cdot poly(|x|)$ for some computable function $f$.
\end{definition}

\begin{lemma}\label{lema:tlm-compute-para-reduction}
	Any parameterized reduction with input instance $(x,k)$ and computable functions $f,g$ can be computed by a $poly(g(k))$-memory $(poly(f(k))\cdot poly(|x|),poly(f(k))\cdot poly(|x|))$-time TM-TLM.
\end{lemma}

\begin{proof}
	By the polynomial time equivalence of Turing machine and other computation models, the parameterized reduction can be computed by a classical Turing machine in $poly(f(k)\cdot poly(|x|))=poly(f(k))\cdot poly(|x|)$ time. Then combining with Theorem \ref{thrm:tlm-sim-tm} which uses TM-TLM to simulate Turing machine, it is proved that the parameterized reduction can be computed by a TM-TLM in $(poly(f(k))\cdot poly(|x|),poly(f(k))\cdot poly(|x|))$-time. The main memory size of the TM-TLM is set to $poly(g(k))$ to support computing $g(k)$ in main memory.
	\qed
\end{proof}

\begin{theorem}\label{thrm:fpt-io-simulate}
	For a problem $\mathcal{P}$ that is FPT with parameter $k$, i.e., there is a $f(k)poly(n)$-time Turing machine that solves $\mathcal{P}$, then there exists a $poly(g(k))$-memory $(g(k)^2n^{g(k)}, g(k)^2n^{g(k)})$-time TM-TLM that can solve $\mathcal{P}$ , where $g$ is a sufficiently large computable function.
\end{theorem}

\begin{proof}
	Since the $k$-clique problem is W[1]-hard and the class of FPT$\subseteq$W[1], then for every problem $\mathcal{P}$ that is FPT, there exists a parameterized reduction from $\mathcal{P}$ to $k$-clique. Then we have the following process to solve any problem $\mathcal{P}$ that is FPT. First use the parameterized reduction to reduce $\mathcal{P}$ to $k$-clique, using a $poly(g(k))$-memory $(poly(f(k))poly(n),poly(f(k))poly(n))$-time TM-TLM, as described in Lemma \ref{lema:tlm-compute-para-reduction}. The obtained instance of $k$-clique has a parameter of $g(k)$ where $g$ is a sufficient large computable function. By Lemma \ref{lema:tlm-compute-kclique}, this instance of $k$-clique can be solved using a $(g(k)+1)$-memory TM-TLM in $(g(k)^2n^{g(k)}, g(k)^2n^{g(k)})$-time. In conclusion, the whole above process can be done with a $poly(g(k))$-memory $(g(k)^2n^{g(k)}, g(k)^2n^{g(k)})$-time TM-TLM.
	\qed
\end{proof}

\section{The Random Access Turing Machine with two-level memory}\label{sec:ratm-tlm}
Since the TM-TLM can not be used to analyze sub-linear time complexity, here we propose the Random Access Turing Machine with two-level memory (RATM-TLM), which grants TM-TLM with the ability of random access on the main memory. It is also an extension of the RATM given in \cite{gao2020recognizing}.

\begin{definition}
	A RATM-TLM is a system
	$\{Q,\Sigma, \Gamma, M, \delta, Read,Write,B, q_r,\\q_w, q_a,q_f\}$ where\\
	$Q$ is the finite set of states,\\
	$\Sigma$ is the finite set of input symbols,\\
	$\Gamma \supseteq \Sigma$ is the finite set of tape symbols,\\
	$M$ is the size of the finite main memory tape,\\
	$\delta $ is the transition function: $Q\times \Gamma \rightarrow Q\times \Gamma \times\{L,S,R\}$,\\
	$Read\subset \delta$ is the read operations,\\
	$Write\subset \delta$ is the write operations,\\
	$B\in \Gamma\setminus \Sigma$ is the blank symbol,\\
	$q_r\subset Q$ is the read states,\\
	$q_w\subset Q$ is the write states,\\
	$q_a\in Q$ is the random access state, and\\
	$q_f\in Q$ is the accepting state.
\end{definition}

According to the idea that RATM-TLM is a combination of random access on the main memory tape and IO operations on the external tape, we have the following corollaries from the results in Section \ref{sec:tm-tlm} and in \cite{gao2020recognizing}.

\begin{corollary}\label{coro:RATM-TLM-relation}
	If $f$ is computable by a RATM in $T(n)$ time , then it is computable by a $(T(n),T(n))$-time RATM-TLM.
	
	If $f$ is computable by a $(T(n),IO(n))$-time RATM-TLM, then it is computable by a RATM in $O(T(n)+IO(n))$ time.
\end{corollary}

\begin{corollary}
	If $f$ is computable by a  $(T(n),IO(n)$-time RATM-TLM, then it is computable by a DTM in $O((T(n)+IO(n))^2\log{(T(n)+IO(n))})$ time.
\end{corollary}

\section{The Random Access Turing Machine with Blocking-IO}\label{sec:RATM-BIO}
\subsection{The definition of Random Access Turing Machine with Blocking-IO}
The TM-TLM and RATM-TLM consider the IO operations as oracles, which assumes that the cost of performing IO operations and main memory computations are equal. However, in realistic computers the IO operations usually take more time than CPU computations. Therefore, in this section the cost of IO operations are modeled explicitly, which provides a theoretical point of view on IO-efficient algorithms.

The Random Access Turing Machine with Blocking-IO (RATM-BIO) is a Random Access Turing machine that can switch between computation states and IO states.
There are two heads in the RATM-BIO, one for main memory computation and the other for external memory data access. They are called the main head and the external head respectively. There are four tapes in RATM-BIO, which are the main memory tape and external memory tape, and the address tapes for the two memory tapes. The main memory head is granted the ability of random access, while the external head must moves consecutively on the external tape like a classical Turing machine.

The transitions of RATM-BIO are grouped into three sets, which are the main memory computation transitions, the external memory access transitions, and read/write transitions. For the main memory computation transitions, the main head moves on the main memory tape and the external head halts. For the external memory access transitions, the main head halts and the external head moves on the external tape. This behavior reflects the blocking IO.

The main memory computation transitions and the external memory access transitions alternate via read/write operations in the following way. When entering the read(write) state, the main head will write the address of the cell in the external tape to the address tape, then the main head halts and the external head starts to move on the external tape to the designated cell. In such way, the RATM-BIO leaves the main memory computation states and enters the external memory access states. When the designated cell is reached, the content in it will replace (be replaced by) the content in the cell where the main head points to, then the external head halts and the main head starts to move again, i.e., the RATM-BIO leaves the external memory access states and enters main  memory computation states. The formal definition of RATM-BIO is given below.

\begin{figure}[t]
	\centering
	\includegraphics[width=0.8\linewidth]{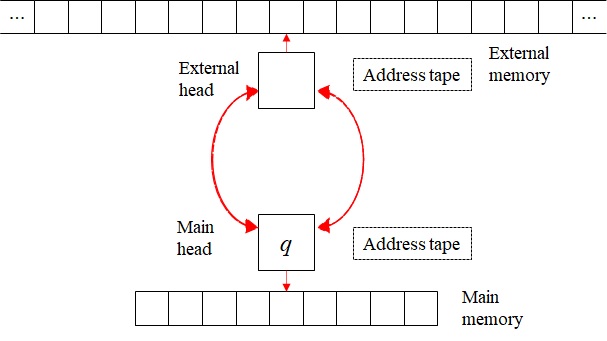}
	\caption{RATM-BIO}\label{fig:ratm-bio}
\end{figure}

\begin{definition}\label{def:ratm-bio}
	The Random Access Turing Machine with Blocking-IO (RATM-BIO) is a system $\{Q,Q_m,Q_e,\Sigma, \Gamma, M, \delta_m,\delta_e, Read,Write, B, q_f,q_r,q_w,q_a\}$ where\\
	$Q$ is the finite set of states,\\
	$Q_m\subset Q$ is the finite set of main memory computation states,\\
	$Q_e\subset Q$ is the finite set of external memory access states,\\
	$\Sigma$ is the finite set of input symbols,\\
	$\Gamma \supseteq \Sigma$ is the finite set of tape symbols,\\
	$M$ is the size of the finite main memory tape,\\
	$\delta$ is the transition function: $\delta \subseteq Q\times \Gamma \rightarrow Q\times \Gamma \times\{L,S,R\}$,\\
	$\delta_m\subset \delta$ is the main memory computation transitions,\\
	$\delta_e\subset \delta$ is the external memory access transitions,\\
	$Read \subset \delta$ is the read operations,\\
	$Write \subset \delta$ is the write operations,\\
	$B\in \Gamma\setminus \Sigma$ is the blank symbol,\\
	$q_r\subset Q$ is the read states,\\
	$q_w\subset Q$ is the write states,\\	
	$q_a\in Q$ is the random access state, and\\
	$q_f\in Q$ is the accepting state.\\
\end{definition}

\subsection{The external access trace complexity}
\begin{definition}[external access trace complexity]
	The \textit{external access trace complexity} of the RATM-BIO is defined to be the total number of moves of the external head. 
\end{definition}

\begin{theorem}\label{thrm:trace-lb}
	For a $M$-memory $(T(n),IO(n))$-time RATM-TLM $\mathcal{M}$, a $M$-memory RATM-BIO $\mathcal{M}'$ can simulate $\mathcal{M}$ with $\Omega(IO(n))$ external access trace complexity.
\end{theorem}
\begin{proof}
=	This lower bound is obtained by doing IO operations on consecutive external cells.
	\qed
\end{proof}

\begin{theorem}\label{thrm:trace-ub}
	For a $M$-memory $(T(n),IO(n))$-time $S(n)$-space RATM-TLM $\mathcal{M}$, a RATM-BIO $\mathcal{M}'$ can simulate $\mathcal{M}$ with $O(IO(n)\cdot S(n))$ external access trace complexity.
\end{theorem}

\begin{proof}
	The upper bound on the external access trace complexity is obtained by the following extreme situation. For each IO operation of $\mathcal{M}$, the external head passes through $O(S(n))$ cells to get the destination cell. 
	\qed
\end{proof}

According to theorem \ref{thrm:trace-lb} and \ref{thrm:trace-ub}, the external access trace complexity of an external algorithm varies significantly using different access pattern of the external tape. It is important to make IO operations on consecutive cells, rather than random accessing. This is widely acknowledged in the research of IO-efficient algorithms, and the above two theorems give a theoretical evidence. 

\section{Usage of proposed models}\label{sec:usage}
\subsection{Usage of TM-TLM}
Here we provide an example to show the usage of TM-TLM on solving FPT problems.

\begin{example}
	The $k$-Vetex Cover problem is known to be FPT. Consider the following fixed-parameter tractable algorithm for $k$-vertex cover. Enumerate all possible subsets of the vertex set with size $k$, then verify whether the subset is a valid vertex cover. This algorithm can be simulated on a $(k+1)$-memory TM-TLM in $(kn^k,kn^k)$-time  as follows. Enumerate the subset in the main memory using $k$ cells, and read the edges from the external memory iteratively into the last cell of the main memory to verify the vertex cover. Note that the main memory usage is less than $poly(g(k))$, and the time and IO complexity is actually less than $O(g(k)^2n^{g(k)})$ given in Theorem \ref{thrm:fpt-io-simulate}, where $g(k)=k$ here.
\end{example}

\subsection{Usage of RATM-TLM}
Now we present an example that needs both the ability of random access in main memory and external memory to achieve sub-linear time complexity and sub-linear IO complexity, which is searching on an external $B^+$-tree.

\begin{example}
	We present a $k$-memory, $(\log_2{k}\log_k{n},k\log_k{n})$-time TM-TLM $\mathcal{M}$ that simulates the procedure of searching on a $B^+$-tree, where $k$ is the size of a node of the $B^+$-tree. The input is a collection $\mathcal{D}$ of data, and a $B^+$-tree $T$ built for the data, which are stored in the external memory of $\mathcal{M}$ at the beginning of computation. Denote $|\mathcal{D}|=n$, then the height of the $B^+$-tree is $O(\log_k{n})$, and the total size of the $B^+$-tree is $O(n)$. 
	
	When the computation starts, $\mathcal{M}$ reads the first $k$ cells of the $B^+$-tree which are the search keys stored in the root node, using $k$ IO operations. Then $\mathcal{M}$ conducts binary search in the main memory and finds the address of the next node to read. Note that the binary search in main memory must use the random access ability, which can not be achieved by classical Turing machine transitions. This search procedure is recursively executed until a leaf node is reached. Finally $\mathcal{M}$ gets the address of the data from the leaf node of the $B^+$-tree and use it to retrieve the data stored in the external memory. It can be verified that the number of IO operations is $O(k\log_k{n})$, and the number of main memory operations is $O(\log_2{k}\log_k{n})$.
\end{example}

\subsection{Usage of RATM-BIO}
The problem studied  on RATM-BIO is the well known depth-first-search (DFS) on directed graphs. We adopt the semi-external setting, where the size of the main memory is set to $O(|V|)$. Algorithm \ref{alg:dfs-random-access} gives an main memory version of DFS.

\begin{algorithm}[H]
	\caption{DFS-Random Access}\label{alg:dfs-random-access}
	\KwIn{Graph $G=(V,E)$}
	\ForEach{vertex u}{free[u]=1\;}
	\ForEach{vertex u}{
		head=0\;
		stack[head]=u\;
		\While{head$\ge0$}{
			v=stack[head]\;
			head--\;
			\If{free[v]}{
				free[v]=0\;
				\ForEach(\tcp*[f]{Random access IO}){neighbor w of v}{\label{line:dfsr:random-io}
					\If{free[v]}{
						head++\;
						stack[head]=w\;
					}
					
				}	
			}
		}
	}
\end{algorithm}

To execute Algorithm \ref{alg:dfs-random-access} on  RATM-BIO, we need two tapes for the main memory computation. One tape stores the $free[v]$ array using $|V|$ cells, and the other tape simulates the stack. It can be verified that the size of the stack does not exceed $O(|V|)$, which satisfies the requirement of semi-external setting.

\begin{theorem}
	If the edges of the input graph $G=(V,E)$ is given on the external tape as adjacent lists, then the worst case external access trace complexity of Algorithm \ref{alg:dfs-random-access} on RATM-BIO is $O(|V||E|)$.
\end{theorem}

\begin{proof}
	According to Algorithm \ref{alg:dfs-random-access}, the only IO operation is incurred by Line \ref{line:dfsr:random-io}, which reads the adjacent list of each node in an unpredictable order. In such way, there always exists a bad access order of the nodes and a bad allocation of the adjacent lists on the external tape, which forces the external head to move for a distance of $\Theta(|E|/2)$ to access the next adjacent list.  Figure \ref{fig:demo-bad-access} demonstrates the idea.
	\begin{figure}[t]
		\centering
		\includegraphics[width=0.7\textwidth]{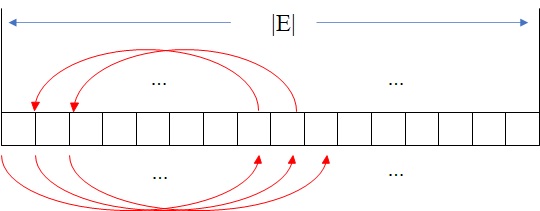}
		\caption{Demonstration of a bad access pattern}\label{fig:demo-bad-access}
	\end{figure}
	
	Thus, the worst case external access trace complexity is $\sum\limits_{v\in V} \{d_v +|E|/2\}=O(|V||E|)$.
	\qed
\end{proof}

\begin{theorem}
	If the edges of the input graph $G=(V,E)$ is given on the external tape in an \textit{arbitrary order}, then the worst case external access trace complexity of Algorithm \ref{alg:dfs-random-access} on RATM-BIO is $O(|E|^2)$.
\end{theorem}

\begin{proof}
	Given the situation that the edges are stored in an arbitrary order, the neighbors of node $v$ can not be retrieved by consecutive IO operations when executing Line \ref{line:dfsr:random-io}. After reading one edge, The external head must move a long distance to get the next edge. Using the same idea in proving the above theorem, there always exists a bad access order to the edges and a bad allocation of the edges on the external tape, which forces the external head to move for distance of $\Theta(|E|)$ between two read operation. Multiplying $|E|$ read operations and $O(|E|)$ moves for each read operation, the result of this theorem is proved.
	\qed
\end{proof}

The above two theorems show the importance of data arrangement on the external memory to avoid random access. It also reveals the importance to design IO-efficient algorithms, whose performance is not affected by the data allocation on the external memory. Here we cite the following algorithm from \cite{sibeyn2002heuristics} given in Algorithm \ref{alg:dfs-edgebyedge}, which is a good example of IO-efficient algorithm. The performance of it is irrelevant of the allocation order of the edges.

\begin{algorithm}[t]
	\caption{DFS-IO-efficient}\label{alg:dfs-edgebyedge}
	\SetKwProg{Procedure}{Procedure}{\string:}{end}
	\KwIn{Graph $G=(V,E)$, memory size $M$}
	
	Initialize the spanning tree $T$ of graph $G$\;
	$\textit{update}\leftarrow \textbf{true}$\;
	\While{\textit{update}}{
		$T, \leftarrow {Restructure(G,T,M)}$\;
	}

	\Return $T$\;
	\Procedure{Restructure($G$, $T$, $M$)}{	$\textit{update} \leftarrow \textbf{false}$\;
		\For {any edge $(u,v)$ in $G$} {
			\If{$(u,v)$ is forward-cross edge}{
				$\textit{update} \leftarrow \textbf{true}$\;
				$w \leftarrow$ the parent of $v$ in $T$\;
				delete edge $(w,v)$ from $T$\;
				add edge $(u ,v)$ iton $T$\;
			}
		}
		\Return $(T,M)$\;
	}  
	
\end{algorithm}

\begin{theorem}\label{thrm:edgebyedge-scan-time}
	Even if the edges of the input graph $G=(V,E)$ is given on the external tape in an arbitrary order, the worst case external access trace complexity of Algorithm \ref{alg:dfs-edgebyedge} on RATM-BIO is $O(|V||E|)$.
\end{theorem}

\begin{proof}
	Then it only need to prove the maximal number of times for scanning the edges. The depth of the initial spanning tree $T$ is at least 1, and we claim that after one full scan of the edge set, the depth grows at least 1. The detailed proof is omitted. Since the depth of the spanning tree is at most $|V|$, the edges are scanned for at most $|V|$ times. Thus the theorem is proved.
	\qed
\end{proof}

\section{Conclusion}\label{sec:conc}
In this paper we proposed three computation models to analyze the IO complexity deeply. The three models are TM-TLM, RATM-TLM and RATM-BIO, which model the behavior of the main memory and external memory in different granularity. Based on TM-TLM and RATM-TLM, the relationship between IO complexity and other complexity measures are deeply studied. Besides, the external access trace complexity defined based on RATM-BIO can reflect the different cost of different external memory access pattern, and provides a theoretical point of view about the IO-efficient algorithms. These new results provide a deep look into the IOt complexity, and open a new way to study the IO complexity. 

\bibliographystyle{splncs04}
\bibliography{library}

%%
%% Bibliography
%%

%% Please use bibtex, 

\end{document}